\documentclass[11pt, authoryear]{llncs}
\usepackage{natbib}
\usepackage{amsfonts, amsmath}
\usepackage[algo2e, ruled, vlined]{algorithm2e}
\usepackage{xspace}
\usepackage{varwidth}
\usepackage{tikz}
\usepackage{url}
\usepackage{bm}
\usepackage{enumerate}
\usepackage{fullpage}
\usepackage[titletoc,title]{appendix}

\newcommand{\E}{\mathop{\mathbb E}}
\newcommand{\n}{\normalsize}
\newcommand{\s}{\normalsize}

\newcommand{\ignore}[1]{}

\begin{document}
\title{The Strategy of Experts for Repeated Predictions}

\author{Amir Ban \inst{1} \and Yossi Azar \inst{2}  \and  Yishay Mansour \inst{3}}

\institute{Blavatnik School of Computer Science, Tel Aviv University \email{amirban@me.com} \and Blavatnik School of Computer Science, Tel Aviv University \email{azar@tau.ac.il} \and Blavatnik School of Computer Science, Tel Aviv University \email{mansour.yishay@gmail.com}}
\maketitle

\date{}

\begin{abstract}
We investigate the behavior of experts who seek to make predictions with maximum impact on an audience. At a known future time, a certain continuous random variable will be realized. A public prediction gradually converges to the outcome, and an expert has access to a more accurate prediction. We study when the expert should reveal his information, when his reward is based on a proper scoring rule (e.g., is proportional to the change in log-likelihood of the outcome).

In \cite{azar2016should}, we analyzed the case where the expert may make a single prediction. In this paper, we analyze the case where the expert is allowed to revise previous predictions. This leads to a rather different set of dilemmas for the strategic expert. We find that it is optimal for the expert to always tell the truth, and to make a new prediction whenever he has a new signal. We characterize the expert's expectation for his total reward, and show asymptotic limits.

\end{abstract}


\section{Introduction}
Situations where a public is interested in the value of a future continuous variable, and has a time-varying consensus estimate of it, are common. Examples abound: Futures and options markets, the weather or climate, results of sport competitions, election results, new book / movie / album sales (for example, the Hollywood Stock Exchange), or economic indicators (for example, Moody's).
We analyze the problem of an expert who makes multiple public predictions in such situations, and in particular, the questions of when to make a first prediction, when to revise a previous prediction, and whether to reveal true beliefs when making a prediction.

Consider, for example, a futures market. A futures market is an exchange where people make contracts to buy specific quantities of a commodity or financial instrument at a specified price with delivery set at a specified time in the future. Traders make money by buying for less than the market's spot price on the delivery date, which we
shall henceforth call the {\em outcome}, or by selling
for more. In effect, a futures market is a prediction market for the
outcome.

The expert is not a trader himself, but someone who is reputed to have access to a more
accurate signal than possessed by regular traders. Often, his
reputation and living is based on this. Stock market analysts,
investment gurus and various types of journalists fit this
description.

The expert contributes to the market by making a public prediction, and is {\em post factum}
rewarded for it. Such
a prediction is a significant market event: Clearly, a market should heed an expert
 whose prediction already encompasses all current common knowledge and adds to it.
We shall below argue that proper scoring rules, and in particular the {\em logarithmic scoring rule}, are the right incentive for the prediction scenario described. Whether the expert's reward takes the form of actual payment, or less tangibly in a boost to his reputation as an expert, is
immaterial to our discussion. We assume that the expert's level of expertise, which we measure by
{\em quality} and describe below, is known to the market. 

We investigate the expert's strategy in such a prediction market. The strategy consists of
choosing the timing and truthfulness of his predictions. Our treatment is Bayesian, assuming all agents draw all possible inferences from their information. In \cite{azar2016should}, we analyzed the case where the expert is allowed a single prediction. In this paper, we study the case of multiple predictions, where an expert is allowed to revise his previous prediction.

\subsection{The Market as a Random Walk}

The  current price in a futures market represents a current
consensus on the outcome (assume that interest rates, or
inflation rates, have been incorporated into the price). According
to the efficient-market hypothesis (EMH), the current price
represents all currently available information, and therefore it is
impossible to consistently outperform the market. Consistent with the
EMH is the random-walk hypothesis, according to which stock market
prices (and their derivatives) evolve according to a random walk and
thus cannot be predicted. By the random-walk hypothesis, the outcome
is the result of a random walk from the current market price.
Equivalently, and the point of view we take in this paper, the
current price is the result of a random walk, reversed in time, from
the outcome (see Figure \ref{random-walk}).

A  random walk adds periodical (say, daily) i.i.d. steps to the
market price. Assuming prices have been adjusted for known trends,
the steps have zero mean. By suitable scaling of the price, the step
variance can be normalized to $1$. Following a common assumption
that the random walk is Gaussian, or lognormal\footnote{Taking logs transforms a lognormal random walk into a Gaussian one.}, the steps have standard
normal distribution (i.e., $N(0,1)$).

\subsection{Expert Quality}

\begin{figure}[tbp]
\centering
\includegraphics[height=0.25\textheight]{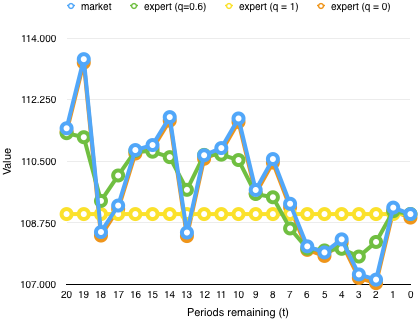}
\caption{Time-dependent signals of a market, a typical expert ($q = 0.6$), a know-all expert ($q=1$), and a know-nothing expert ($q=0$)}
\label{random-walk}
\end{figure}

The expert's expertise consists of having a more accurate signal of
the outcome price $x_0$ than the market's, and the expert's
quality measures by how much. The quality $q \in [0,1]$ measures
what part of the market's uncertainty the expert ``knows'', so that
it does not figure in the expert's own uncertainty. Equivalently,
the expert's uncertainty is $1-q$ of the market's uncertainty. This
proportion is statistical: It is the uncertainties' variances,
rather than their realizations, that are related by proportion. If
the market price is a Gaussian random walk from the outcome with
$N(0,1)$ steps, the expert's prediction is a Gaussian random walk
from the outcome with $N(0,1-q)$ steps.

The  expert's knowledge, i.e., the part of the market's uncertainty
that the expert is not uncertain about, has steps of zero mean and
$q$ variance. On the assumption that the expert's knowledge steps
and uncertainty steps are mutually independent, their sum has the
sum mean and sum variance of their parts, i.e., they sum back to the
market's uncertainty steps of zero mean and variance $q + (1-q) =
1$.

Figure \ref{random-walk} illustrates the evolution of a market's signal in the last 20 periods
until the outcome ($109$) becomes known. Also shown are the private signals of 3
experts predicting the same event, with qualities of $0.6, 1$ and $0$.

An  expert with $q=1$ has no uncertainty at all, and his signal
equals the outcome $x_0$ at all times $t$. At the other extreme, a
(so-called) expert with $q=0$ has no knowledge beyond common
knowledge, and his signal equals the market value $x_t$ at all $t$.

In  this paper an expert's quality is common knowledge, shared by
all traders as well as himself. Whether its value $q$ represents
 objective reality, or is a belief, based, e.g., on past
performance, makes no difference to our discussion.

\subsection{Scoring a Prediction}

\newcommand{\pnow}{\bar{p}}

A scoring rule is a way to evaluate and reward a prediction of a stochastic
event, when that prediction is presented as a distribution over possible results.
 The predictor declares at time $t>0$ a probability
distribution $p\in \Delta(R)$, and at time $0$ some $r\in R$ is
realized. A {\em scoring rule} $S$ rewards the predictor
$S(p,r)$ when his prediction was $p$ and the realized value
is $r$. In market settings, and many other settings, there exists a current prediction $\pnow$
and then the predictor is evaluated on the scoring difference effected $S(p,r)-S(\pnow,r)$. Note that the optimization problem of the predictor in a market situation is the same, since he has no
influence over $S(\pnow,r)$, the only difference is that now the
predictor might be penalized for making the current prediction less accurate.
A {\em proper} scoring rule is a scoring rule for which reporting the
true distribution is optimal according to the predictor's
information.

The logarithmic scoring rule,  with $S(p,r) = \log p_r$ (where $p_r$ is the value of PDF $p$ at $r$), scores a prediction by the
log-likelihood of the outcome according to the prediction. It is proper, and has strong roots in
information theory: In reference to a current prediction $\pnow$,
it scores $\log p_r / \pnow_r$, which, in information theory, is the {\em self-information}, also called {\em surprisal}, contained in the outcome. Conditional on $p$ being the correct distribution, the expected score is the Kullback-Leibler divergence between $p$ and $\pnow$:
$E_{r\sim p} [\log p_r/\pnow_r]= D_{KL}(p || \pnow)$.

In our model expert predictions are scored with the logarithmic scoring rule, which the expert seeks to maximize. This is justified by the following
\begin{itemize}
\item The reward is incentive compatible, eliciting truth-telling by the expert. This enables a Bayesian market to adopt predictions verbatim. A reward that is not incentive compatible would greatly complicate the Bayesian interpretation of predictions, possibly even making our problem indeterminate.
\item In our model (the essential details of which were already sketched), the entire prediction distribution follows from the prediction mean by common knowledge. Since only proper scoring rules are incentive compatible with predictions phrased as distributions over results, it follows that the reward {\em must} be by a proper scoring rule.
\item The logarithmic scoring rule is favored by its unique information-theory meaning, and other unique attributes (e.g., its locality). It is commonly used in real-world predictions markets, in a mechanism called LMSR ({\em Logarithmic Market Scoring Rule}) introduced by \cite{Hanson03}. \cite{Chen2010} say ``LMSR has
become the de facto market maker mechanism for prediction markets.
It is used by many companies including Inkling Markets, Consensus
Point, Yahoo!, and Microsoft''.
\end{itemize}

Proper scoring rules are {\em myopically} incentive compatible for risk-neutral agents, i.e. they are guaranteed to elicit the truth, but only when future actions are not taken into account (or, when there are no future actions, i.e., at the last prediction). As will be further discussed, when future actions {\em are} taken into account, incentive compatibility is {\em not} guaranteed.

\subsection{The Expert's Dilemma}

Assume that the expert has no obligation to speak at any particular
time, or at all. The reward for no prediction is zero, and for each prediction made,
the expert is rewarded by the logarithmic scoring rule.
The expert may revise his previous prediction by making a new one
whenever this is advantageous. The expert faces several dilemmas: When to make
the first prediction? and when is it appropriate to revise a previous prediction? Moreover,
as proper scoring rules are incentive
compatible only with the {\em last} prediction,
is there a strategy more profitable than always telling the truth?

In \cite{azar2016should}, we analyzed the single-prediction case,
and argued that an expert may pass on making a prediction in the hope of getting
a better opportunity later. In the multiple predictions scenario, there is no need
to pass, since the opportunity to make a future prediction remains. Conceivably, the
expert will want to revise his prediction whenever he gets a fresh signal (we find that
this is so), or, he may want to do so only when the new signal significantly changes his
prediction (we find otherwise).

Should he always tell the truth? Whenever the expert makes two or more predictions, he
may conceivably distort his first prediction, hoping to misdirect a gullibly-Bayesian market,
and reaping a net profit by subsequently setting the market right.

\subsection{Summary of Results}
\label{results}

Our results are satisfyingly tidy: Despite apparent temptation to mislead, it is optimal for the expert to always tell the truth, and therefore it is rational for the market to take his predictions at face value. The optimal prediction schedule for the expert is to make a new one whenever he has a new signal and is allowed to speak. We show that the expected total reward for all predictions is, asymptotically for large $t$, $\frac{1}{2} q \log t$, proportional to quality ($q$) and to the log of the number of periods left ($t$).

To some, these results, and especially the truthfulness result, would seem straightforward. However, this intuition is false, and not supported by the literature (see below in Section \ref{lit}, \cite{Chen2007} and \cite{chen2016informational}). The following generic example illustrates why.

\begin{example}
There is a market, who gets public signals, and an expert, who gets private signals.

Suppose that at time $t$ the market receives a signal $x_0 + \epsilon$, where $x_0$ is the outcome, and $\epsilon$ is a random variable. At $t - 1$, and (independently) at $t + 1$, the expert receives a signal $x_0 + \epsilon$ with probability $1/2$, and $x_0 - \epsilon$ with probability $1/2$. 

The expert makes a prediction at both times. Should he reveal his true information?

Whoever sees two different truthful signals is able to calculate the outcome $x_0 = (x_0 + \epsilon)/2 + (x_0 - \epsilon)/2$ exactly.

For any scoring rule, and any distribution of $\epsilon$, the expert should not tell the truth on his first prediction. This prevents the $50\%$ probability that the market will know $x_0$ at $t$,  preserving a $75\%$ probability that the expert can announce $x_0$ on his second prediction.
\end{example}

\subsection{Related Literature}
\label{lit}

Learning from expert opinion and its aggregation has a long history, with \cite{degroot1974reaching} and the Bayesian framework by Morris (\cite{morris1974decision}, \cite{morris1977combining}) leading to much subsequent work. While much of this work treats experts as oracles with no motivation of their own, some of it took a look at an expert's concern for his reputation, i.e., the wish to appear well-informed. In \cite{bayarri1989optimal} the setting was a weighted averaging of several expert opinions, with the weights adjusted by observed accuracy when the outcome is known. Experts wish to maximize their posterior weight. The authors found that incentive compatibility is attainable only by assigning a logarithmic utility to the weight. In \cite{ottaviani2006strategy} and \cite{ottaviani2006reputational} the authors cast the expert's inferred type as the optimization target. Their type, a real number, is a cognate of our quality. The authors argue that truth-telling is generally not possible, as experts are motivated to simulate better quality than they actually have. In our model, expert quality is common knowledge, and so not open to manipulation.

The question of timing has received attention from \cite{kreps1978temporal}, who laid out a basis for the temporal resolution of uncertainty in dynamic choices.  \cite{ottaviani2001information} discuss the optimal order of speaking to avoid herding in a committee where members have heterogenous expertise.

\cite{Chen2007} as well as \cite{chen2016informational} studied situations where several agents, each having private information, are given more than one opportunity to make a public prediction. The canonical case is ``Alice-Bob-Alice'' where Alice speaks before and after Bob's single speaking opportunity, both are awarded by a proper scoring rule for each prediction, and both maximize their total score. The proper scoring rule assures that each will tell the truth on their last prediction, and the open question is whether Alice, when going first, will tell the truth, lie, or keep her silence. \cite{Chen2007} show cases where Alice is motivated to mislead on her first prediction, and make the key observation that truthfulness is optimal if, in a different setup, namely, a single-prediction Alice-Bob game where Alice chooses whether to go first or second, she will always prefer going first. Building on that insight, \cite{chen2016informational} show that when the players' information is what they define as ``perfect informational substitutes'', they will predict truthfully and as early as allowed, while when they are ``perfect informational complements'', they will predict truthfully and as {\em late} as allowed, while when players are neither substitutes nor complements, untruthfulness can and will occur.

These works differ from ours in that they model agents having a constant piece of information, which they may choose when to reveal, while we model agents (expert and market) as receiving a time series of signals with new information every time period. (And the martingale property of random walks assures that the new information cannot be predicted from old). In the Discussion we comment on how our results reflect on a possible generalization of the mentioned works to dynamic-information settings.


The  Efficient Market Hypothesis was introduced by \cite{Fama1969}. The Random Walk Hypothesis is even older, originating in the 19th century, and discussed by, e.g., \cite{Samuelson1965} and \cite{Fama1965}. The Black-Scholes option pricing model \cite{Black} is based on a Gaussian random walk assumption.

Scoring rules have a very long history, going back to \cite{Finetti37}, \cite{Brier50} and \cite{Good52}. Proper scoring rules are often used for incentive-compatible belief elicitation of risk-neutral agents (e.g. \cite{armantier2013eliciting}).

\subsection{Paper Organization}

The  rest of this paper is organized as follows: In Section
\ref{model} we describe our model. After establishing some preliminary results in Section \ref{prelim},
Section \ref{multiple} is devoted to the multiple-prediction problem.
In Section \ref{discussion} we summarize and
offer concluding remarks.

\section{Model}
\label{model}

\subsection{Market Prediction}

A market predicts the outcome of a continuous random variable $X_0$,
whose realized value $x_0$ will be revealed at time $0$. Time is
discrete and flows backwards from an initial period $T_{max}$, i.e., $T_{max},
\ldots, t, \ldots, 1, 0$. At any time $t>0$ the market observes
$X_0+{\mathcal Z}_t$ where ${\mathcal Z}_t\sim N(0,t)$. We model ${\mathcal
Z}_t$ as a random walk with i.i.d. steps $Z_t, \ldots,
Z_1$, i.e., ${\mathcal Z}_t=\sum_{\tau=1}^t Z_\tau$ and $Z_\tau\sim
N(0,1)$. Let the market prediction (when uninformed by experts) be $X_t
:= X_0 + {\mathcal Z}_t$ at time $t$, and let $x_t$ be the realized
value. With every passing period $t$, the value of $Z_t = z_t$ is
revealed and becomes common knowledge, and the market's new
prediction changes to $x_{t-1} = x_t - z_t$. Note that the variance
of ${\mathcal Z}_t$ decreases with time, and at
time $0$ the market's prediction coincides with the outcome $x_0$.
The random variable $X_0$ is normally distributed $N(0,\sigma_0^2)$ where we assume
$\sigma_0^2 \gg T_{max}$. This assumption means that the outcome is, for practical purposes, unconstrained
by a prior, and makes posterior computations
dependent solely on observed signals, since\footnote{When a normal
variable with prior distribution $N(0,\sigma_0^2)$
is sampled with  known variance $t$ at value $x_t$, its Bayesian posterior
distribution is normal with mean $\frac{x_t/t}{1/\sigma_0^2+1/t}$
and variance $\frac{1}{1/\sigma_0^2+1/t}$. Assuming $\sigma_0^2 \gg
T_{max} \geq t$, this simplifies to $N(x_t,t)$.} we have $\E[X_0|X_t=x_t]=x_t$ and
$Var(X_0| X_t=x_t)=t$.

\subsection{Expert Information and Goal}

There is an expert, with quality $q \in [0,1]$, whose quality is
common knowledge. The expert's quality consists in ``knowing'' part
of the random steps $Z_t$ of every period, and therefore getting a
more accurate signal of $X_0$. Formally,
\begin{itemize}
\item For every $t$, $Z_t = A_t + B_t$, where $A_t \sim N(0,q)$ and $B_t\sim N(0,1-q)$ are mutually independent.
(Note that $Z_t\sim N(0,1)$.)
\item The expert's private signal at time $t$ is $Y_t = X_0 + B_1 + \ldots + B_t$ and let $y_t$ be its realized value. (Note that if $q=0$ then $Y_t=X_t$ and if $q=1$ then $Y_t=x_0$.)
\end{itemize}

At every $t > 0$, the expert may choose to make a prediction of the outcome. The market has a varying probability distribution on the outcome,
which is affected by its signals and by the expert's predictions. (The market price is
the distribution mean). Each prediction is scored by the logarithmic scoring rule as described
below. The expert's reward is the total score for all predictions
made.

The expert's outcome distribution at $t$ is $N(y_t, (1-q) t)$.\footnote{Since the expert is better informed than the market, his prediction depends on his signal alone. This is formally proved in Proposition \ref{at-prediction}.}
In practice, it is enough for the expert to announce $y_t$ as his entire distribution follows
by the model and common knowledge. A prediction's reward is determined
at time $0$ based on the realized value ($x_0$) by the logarithmic scoring rule.
Namely, if the market distribution prior to the expert prediction is $X_{t-}$ with density $f_{-}$, and following the expert prediction
the posterior market distribution is $X_{t+}$ with density
$f_{+}$, then the expert reward is $\log (f_{+}(x_0)/f_{-}(x_0))$, where $x_0$ is the realized value.

An expert who refrains
from making any prediction is awarded $0$. The expert
optimization problem is to maximize his expected reward given his
private information. {\em The question before the expert is if and
when to make predictions, and, when allowed multiple predictions, whether to be truthful in his predictions.}

\section{Preliminaries}
\label{prelim}

\subsection{Time and Expectation Notation}

Distributions and other variables often use a time subscript, e.g., $X_t$ is market's distribution at $t$ ($t$ periods before delivery date). When a prediction takes place at $t$, the notation $X_{t-}, X_{t+}$ is used to distinguish between the variable before, and after, respectively, the prediction.

We use the notation $\E\limits_t[Z]$ to denote the expectation of a random variable $Z$ according to the distribution known at $t$. This is shorthand for $\E\limits_{X_t}[Z]$ when referring to the market's expectation, or for $\E\limits_{Y_t}[Z]$ when referring to the expert's expectation. Which of the two is meant will either be clear from the context or explicitly stated. If a prediction was made at $t$, we use the notation $\E\limits_{t-}[Z], \E\limits_{t+}[Z]$ to distinguish between the market's expectation of $Z$ before and after, respectively, the prediction has been made.

\subsection{A Criterion for Independence}

Here we prove a result about random variables based on model signals that will enable us to determine whether they are stochastically independent.

From the model definitions we have, for every $i \geq j$
\s
\begin{align}
\label{covxx}
Cov(X_i, X_j) &= j \\
\label{covxy}
Cov(Y_i, X_j) &= Cov(X_i, Y_j) = Cov(Y_i, Y_j) = (1- q)j
\end{align}
\n

\begin{lemma}
\label{joint}
Define $\bm{U}$ to be the random vector \\$(X_0, X_1, Y_1, \ldots, X_t, Y_t)^T$, and let $\bm{U_1}, \bm{U_2}$ be two random vectors of linear combinations of the $\bm{U}$. Then the joint distributions of $\bm{U_1}, \bm{U_2}$ are mutually independent if and only if for every pair $u_1 \in \bm{U_1}$ and $u_2 \in \bm{U_2}$ $Cov(u_1, u_2) = 0$.
\end{lemma}

\begin{proof}
Every linear combination of $\bm{U}$ is normal, as it is a linear combination of the i.i.d. normal variables $A_i, B_i, i = 1, \ldots t$, and of $X_0$, which is normal and independent of each of the others. Therefore $\bm{U}$ has a jointly multivariate normal distribution. Therefore so has random vector $\left( \begin{smallmatrix} \bm{U_1} \\ \bm{U_2} \end{smallmatrix} \right)$. The lemma states a general property of jointly multivariate normal distributions, see \cite{Tong} Theorem 3.3.2.
\qed
\end{proof}

\subsection{Market Effect of a Prediction}
\label{effect-at}

Before evaluating, and then optimizing, expert's value for a prediction strategy, we must resolve how a single prediction affects market price, at the time of prediction. For the multiple-prediction problem, we also need to determine the effect {\em after} the prediction was made. This will be resolved in Section \ref{effect-after}.

Define $\mathcal{Z}_t$ to be the set of all expert and market signals previous to $t$ ($X_\tau, Y_\tau$ for $\tau > t$).

Assume the expert makes a prediction $y_t$ at time $t$, implying a distribution of $Y_t \sim N(y_t, (1-q)t)$. Then at time $t$ the market's posterior distribution will be the expert's announced distribution, regardless of the market's signal at this time $x_t$ and all previous signals $\mathcal{Z}_t$.

\begin{proposition}
\label{at-prediction}
If the expert makes a prediction $y_t$ at time $t$, the market's posterior distribution is the expert's implied distribution $Y_t \sim N(y_t, (1-q)t)$.
\end{proposition}

\begin{proof}
Let $\mathcal{Z} = \mathcal{Z}_t \cup \{X_t\}$. Define $\mathcal{Z}' = \mathcal{Z} - Y_t$ to be the set of random variables $Z - Y_t$ with $Z \in \mathcal{Z}$.

As given $Y_t$ there is a one-to-one correspondence between $\mathcal{Z}$ and $\mathcal{Z}'$
\s\begin{align*}
\Pr[x_0 | Y_t, \mathcal{Z}] = \Pr[x_0 | Y_t, \mathcal{Z}']
\end{align*}\n

By \eqref{covxx}, \eqref{covxy} for each $Z \in \mathcal{Z}'$, $Cov(Z, Y_t) = 0$ and $Cov(Z, X_0) = 0$. So by Lemma \ref{joint}
\s\begin{align*}
\Pr[x_0 | Y_t, \mathcal{Z}'] = \frac{\Pr[x_0, Y_t, \mathcal{Z}']}{\Pr[Y_t, \mathcal{Z}']} = \frac{\Pr[x_0, Y_t]\Pr[\mathcal{Z}']}{\Pr[Y_t]\Pr[\mathcal{Z}']} =  \Pr[x_0 | Y_t]
\end{align*}\n

Therefore, as claimed
\s\begin{align*}
\Pr[x_0 | Y_t, \mathcal{Z}] = \Pr[x_0 | Y_t]
\end{align*}\n
\qed
\end{proof}

\subsection{Prediction Score Expectation}

Assume the expert makes a prediction at time $t$.
Let the market prediction prior to the expert prediction be $X_{t-}
\sim N(\mu_{-},\sigma_{-}^2)$ with density $f_{-}$ and let the posterior
market prediction be $X_{t+} \sim N(\mu_{+},\sigma_{+}^2)$ with density
$f_{+}$. Let expert's reward be denoted by $W$, then

\s\begin{align}
\label{W}
W &= \log \frac{f_{+}(x_0)}{f_{-}(x_0)} = \log
\frac{\frac{1}{\sigma_{+}\sqrt{2\pi}}e^{-\frac{(x_0 -
\mu_{+})^2}{2\sigma_{+}^2}}}{\frac{1}{\sigma_{-}\sqrt{2\pi}}e^{-\frac{(x_0
- \mu_{-})^2}{2\sigma_{-}^2}}} \nonumber \\
&= \log\frac{\sigma_{-}}{\sigma_{+}} + \frac{(x_0 - \mu_{-})^2}{2\sigma_{-}^2} - \frac{(x_0 - \mu_{+})^2}{2\sigma_{+}^2}
\end{align}\n
As the reward depends on $x_0$, its value is only known at time $0$. The expert can calculate his reward expectation when making it (at $t$), based on his belief about the distribution of $x_0$.

Consider the case that the expert prediction is truthful, so that by Proposition \ref{at-prediction} after it his beliefs of $x_0$ are identical to the market's.

\begin{proposition}
\label{expected}
If the market's prediction before an expert prediction is $X_{t-} \sim N(\mu_{-},\sigma_{-}^2)$, and after an expert prediction is $X_{t+} \sim N(\mu_{+},\sigma_{+}^2)$, then if the prediction is truthful the expert's expected reward is positive and equals the Kullback-Leibler divergence $D_{KL}(X_{t+} || X_{t-})$.

\s\begin{align}
\label{E_W}
\E_t[W] &= D_{KL}(X_{t+} || X_{t-}) =
\frac{(\mu_+ - \mu_-)^2}{2\sigma_-^2} + \frac{1}{2}\Bigl(\frac{\sigma_+^2}{\sigma_-^2} - 1 - \log \frac{\sigma_+^2}{\sigma_-^2}\Bigr)
\end{align}\n
\end{proposition}

\begin{proof}
As the second moment of the normal distribution $N(\mu,\sigma^2)$ is $\mu^2 + \sigma^2$, and since the expert's distribution translates to
\s\begin{align*}
x_0 - \mu_{-} &\sim N(\mu_{+} - \mu_{-},\sigma_{+}^2) \\
x_0 - \mu_{+} &\sim N(0,\sigma_{+}^2),
\end{align*}\n
\noindent we get by taking expectations from \eqref{W}
\s\begin{align*}
\E_t[W] = \E_{x_0\sim N(\mu_{+},\sigma_{+}^2)}[W] &= \log\frac{\sigma_{-}}{\sigma_{+}} + \frac{(\mu_{+} - \mu_{-})^2 + \sigma_{+}^2}{2\sigma_{-}^2}- \frac{0 + \sigma_{+}^2}{2\sigma_{+}^2} \\ 
&= \frac{(\mu_+ - \mu_-)^2}{2\sigma_-^2} + \frac{1}{2}\Bigl(\frac{\sigma_+^2}{\sigma_-^2} - 1 - \log \frac{\sigma_+^2}{\sigma_-^2}\Bigr)
\end{align*}\n

This is positive, because for every $x < 1$, $\log (1-x) \leq -x$
\qed
\end{proof}

We apply the above result to calculate the expected reward of a {\em first} prediction.

\begin{proposition}
\label{first}
For an expert's first prediction at $t$, his reward expectation is
\s\begin{align}
\label{KLdivergence}
\E_t[W] =  \frac{(x_t - y_t)^2}{2t} - \frac{1}{2}\Bigl(q + \log (1 - q)\Bigr)
\end{align}\n
\end{proposition}

\begin{proof}
For an expert's first prediction, we have $\mu_- = x_t$,
$\sigma_-^2 = t $. By Proposition \ref{at-prediction} $\mu_+ = y_t$, and $\sigma_+^2 = (1 - q) t$.
Hence, $X_t^{-} \sim N(x_t,t)$ and $X_t^{+}\sim N(y_t,(1-q)t)$. Substituting these in \eqref{E_W} we
derive \eqref{KLdivergence}.
\qed
\end{proof}
Note that the reward $W$ may be positive or negative
depending on $x_0$, but its expectation is always non-negative.

\section{The Optimal Multiple-Prediction Strategy}
\label{multiple}

\subsection{Market Effect after a Prediction}
\label{effect-after}

When an expert has made a prediction at $T$, what is the market's posterior distribution at the next periods $T-1, T-2, \ldots$, assuming the expert makes no new predictions? This is more complex than at the time of prediction (see Section \ref{effect-at}), and the distribution depends on more than one signal, as stated in the following proposition

\begin{proposition}
\label{after-prediction}
At time $t$, Let $T > t$ be the time of expert's latest prediction $y_T$. Let $\mathcal{Z} := \mathcal{Z}_T \cup \{y_T, x_T,  x_{T-1}, \ldots x_t\}$. Then at $t$ the market's posterior distribution is \\ $N(\mu(x_0 | \mathcal{Z}), Var (x_0 | \mathcal{Z}))$ with
\s\begin{align*}
\mu(x_0 |\mathcal{Z}) &= \frac{\frac{x_t}{t} + \frac{1}{1-q}\frac{y_T}{T} - \frac{x_T}{T}}{\frac{1}{t} + \frac{q}{1-q}\frac{1}{T}}\\
Var (x_0 | \mathcal{Z}) &= \frac{1}{\frac{1}{t} + \frac{q}{1-q}\frac{1}{T}}
\end{align*}\n
\end{proposition}

\begin{proof}

In outline, the proof is based on showing:
\begin{enumerate}
\item The distribution of $Y_t | y_T, x_T, x_t$.
\item $\Pr[Y_t | \mathcal{Z}] = \Pr[Y_t | y_T, x_T, x_t]$, i.e., no observation other than $y_T, x_T, x_t$ affects $Y_t$'s posterior. 
\item How the distribution of $x_0 |\mathcal{Z}$ is derived from that of $Y_t |\mathcal{Z}$.
\end{enumerate}

Our proof will show the posterior distribution of $y_t$ conditional on the same random variables $\mathcal{Z}$, from which the posterior distribution of $x_0$ will follow by the following lemma.

\begin{lemma}
\label{Y_t}
\s\begin{align*}
\mu(x_0 | \mathcal{Z}) &= \mu(Y_t | \mathcal{Z}) \\
Var (x_0 | \mathcal{Z}) &= Var (Y_t | \mathcal{Z}) + (1-q) t
\end{align*}\n
\end{lemma}

\begin{proof}
For each $Z \in \mathcal{Z}$ the value of $Z - Y_t$ is independent of $Y_t - x_0$. Therefore
\s\begin{align*}
\mu(x_0 | \mathcal{Z}) &= \mu(Y_t | \mathcal{Z}) - \mu(Y_t - x_0) \\
Var (x_0 | \mathcal{Z}) &= Var (Y_t | \mathcal{Z}) + Var(Y_t - x_0)
\end{align*}\n

As $Y_t - x_0 = \sum\limits_{i=0}^t B_i \sim N(0,(1-q)t)$, the lemma follows.
\qed
\end{proof}

We will next show the distribution of $Y_t$ conditional only on $y_T, x_T, x_t$. This will be followed by a proof that these are the only variables in $\mathcal{Z}$ that matter.

\begin{lemma}
\label{three-vars}
Conditional only on $y_T, x_T$ and $x_t$, the distribution of $Y_t$ is normal, with
\s\begin{align}
\label{y_t_mean}
\mu(Y_t | y_T, x_T, x_t) &= \frac{\frac{x_t}{t} + \frac{1}{1-q}\frac{y_T}{T} - \frac{x_T}{T}}{\frac{1}{t} + \frac{q}{1-q}\frac{1}{T}}\\
\nonumber
Var (Y_t | y_T, x_T, x_t) &= \frac{1}{\frac{1}{qt} + \frac{1}{(1-q)(T-t)}+\frac{1}{q(T-t)}}
\end{align}\n
\end{lemma}

\begin{proof}
By the model, the observed variables $Y_T, X_T, X_t$ are related to $Y_t$ as follows
\s\begin{align}
\label{X_t}
X_t &= Y_t + \sum\limits_{i=0}^t A_i \\
X_T &= Y_t + \sum\limits_{i=t+1}^T A_i + \sum\limits_{i=t+1}^T B_i  + \sum\limits_{i=0}^t A_i \\
\label{Y_T}
Y_T &= Y_t + \sum\limits_{i=t+1}^T B_i
\end{align}\n

Define a random variable $Z := Y_T - X_T + X_t$, and let $z := y_T - x_T + x_t$ be its observed value. Then, by the above
\s\begin{align}
\label{Z}
Z &= Y_t -  \sum\limits_{i=t+1}^T A_i
\end{align}\n

Now, observing \eqref{X_t}, \eqref{Y_T} and \eqref{Z}, as $\sum\limits_{i=0}^t A_i$, $\sum\limits_{i=t+1}^T B_i$ and $\sum\limits_{i=t+1}^T A_i$ are pairwise independent, $x_t, y_T$ and $z$ are three observations of $Y_t$ with independent ``noises'', each with normal distribution of zero mean, and variances $qt, (1-q)(T-t)$ and $q(T-t)$, respectively. The posterior distribution of $Y_t$ is therefore normal with
\s\begin{align*}
\mu(Y_t | y_T, x_T, x_t) &= \frac{\frac{x_t}{qt} + \frac{y_T}{(1-q)(T-t)} + \frac{z}{q(T-t)}}{\frac{1}{qt} + \frac{1}{(1-q)(T-t)}+\frac{1}{q(T-t)}}\\
Var (Y_t | y_T, x_T, x_t) &= \frac{1}{\frac{1}{qt} + \frac{1}{(1-q)(T-t)}+\frac{1}{q(T-t)}}
\end{align*}\n

Substituting $z = y_T - x_T + x_t$ and simplifying
\s\begin{align*}
\mu(Y_t | y_T, x_T, x_t) &= \frac{\frac{x_t}{qt} + \frac{y_T}{(1-q)(T-t)} + \frac{z}{q(T-t)}}{\frac{1}{qt} + \frac{1}{(1-q)(T-t)}+\frac{1}{q(T-t)}} = \frac{\frac{x_t}{t} + \frac{1}{1-q}\frac{y_T}{T} - \frac{x_T}{T}}{\frac{1}{t} + \frac{q}{1-q}\frac{1}{T}}
\end{align*}\n

\noindent which completes the proof.
\qed
\end{proof}

\begin{lemma}
\label{redundant}
\s$$\Pr(Y_t | \mathcal{Z}) = \Pr(Y_t | X_t, X_T, Y_T)$$\n
\end{lemma}

\begin{proof}
We write, using Bayes' rule and the chain rule
\s\begin{align}
\nonumber
\Pr[Y_t | \mathcal{Z}] &= \frac{\Pr[Y_t, \mathcal{Z}]}{\Pr[\mathcal{Z}]} \\
\nonumber
&= \frac{\Pr[\mathcal{Z}_T, Y_t, X_{t+1}, \ldots, X_{T-1} | X_t, X_T, Y_T]}{\Pr[\mathcal{Z}_T, X_{t+1}, \ldots, X_{T-1} | X_t, X_T, Y_T]} \\
\label{chain-rule}
&= \Pr[Y_t | X_t, X_T, Y_T] \frac{\Pr[\mathcal{Z}_T, X_{t+1}, \ldots, X_{T-1} | X_t, X_T, Y_T, Y_t]}{\Pr[\mathcal{Z}_T, X_{t+1}, \ldots, X_{T-1} | X_t, X_T, Y_T]}
\end{align}\n

Define $\mathcal{Z}' = \mathcal{Z}_T - Y_T$ to be the set of random variables $Z - Y_T$ with $Z \in \mathcal{Z}_T$, and for every $\tau \in [t+1, T-1]$ define
\s\begin{align*}
R_\tau := \frac{\tau - t}{T-t}(X_\tau - X_T) + \frac{T-\tau}{T-t}(X_\tau - X_t)
\end{align*}\n

As given $X_t, X_T$ there is a one-to-one correspondence between $X_\tau$ and $R_\tau$, and given $Y_T$ there is a one-to-one correspondence between $\mathcal{Z}_T$ and $\mathcal{Z}'$
\s\begin{align}
\label{subs1}
\Pr[\mathcal{Z}_T, X_{t+1}, \ldots, X_{T-1} | X_t, X_T, Y_T, Y_t] &= \nonumber \\
\Pr[\mathcal{Z}', R_{t+1}, \ldots, R_{T-1} | X_t, X_T, Y_T, Y_t] \\
\label{subs2}
\Pr[\mathcal{Z}_T, X_{t+1}, \ldots, X_{T-1} | X_t, X_T, Y_T] &= \nonumber \\ 
\Pr[\mathcal{Z}', R_{t+1}, \ldots, R_{T-1} | X_t, X_T, Y_T]
\end{align}\n

By \eqref{covxx}, \eqref{covxy}
\s\begin{align*}
Cov(R_\tau, X_t) &= \frac{\tau - t}{T-t}(t - t) + \frac{T-\tau}{T-t}(t - t) = 0 \\
Cov(R_\tau, Y_t) &= \frac{\tau - t}{T-t}(1-q)(t - t) + \frac{T-\tau}{T-t}(1 - q)(t - t) = 0 \\
Cov(R_\tau, X_T) &= \frac{\tau - t}{T-t}(\tau - T) + \frac{T-\tau}{T-t}(\tau - t) = 0 \\
Cov(R_\tau, Y_T) &= \frac{\tau - t}{T-t}(1-q)(\tau - T) + \frac{T-\tau}{T-t}(1 - q)(\tau - t) = 0
\end{align*}\n

Also, for every $Z \in \mathcal{Z}'$ and every $\tau \leq T$, $Cov(Z, X_\tau) = Cov(Z, Y_\tau) = 0$.

Consequently, applying Lemma \ref{joint}
\s\begin{align*}
\Pr[\mathcal{Z}', R_{t+1}, \ldots, R_{T-1} | X_t, X_T, Y_T, Y_t] &= \Pr[\mathcal{Z}', R_{t+1}, \ldots, R_{T-1}] \\
\Pr[\mathcal{Z}', R_{t+1}, \ldots, R_{T-1} | X_t, X_T, Y_T] &= \Pr[\mathcal{Z}', R_{t+1}, \ldots, R_{T-1}]
\end{align*}\n

Substituting the above in \eqref{subs1}, \eqref{subs2}, \eqref{chain-rule}, $\Pr[Y_t | \mathcal{Z}] = \Pr[Y_t | X_t, X_T, Y_T]$, as claimed.
\qed
\end{proof}

A corollary of Lemma \ref{redundant} and Lemma \ref{three-vars} is that \eqref{y_t_mean} is a sufficient statistic for $y_t$ given the random variables $\mathcal{Z}$.

The proof of the proposition is now easily concluded: From Lemma \ref{redundant} with Lemma \ref{three-vars} we conclude
\s\begin{align*}
\mu(Y_t | \mathcal{Z}) &= \frac{\frac{x_t}{t} + \frac{1}{1-q}\frac{y_T}{T} - \frac{x_T}{T}}{\frac{1}{t} + \frac{q}{1-q}\frac{1}{T}}\\
Var (Y_t | \mathcal{Z}) &= \frac{1}{\frac{1}{qt} + \frac{1}{(1-q)(T-t)}+\frac{1}{q(T-t)}}
\end{align*}\n

Combining the above with Lemma \ref{Y_t},
\s\begin{align*}
\mu(x_0 | \mathcal{Z}) &= \frac{\frac{x_t}{t} + \frac{1}{1-q}\frac{y_T}{T} - \frac{x_T}{T}}{\frac{1}{t} + \frac{q}{1-q}\frac{1}{T}}\\
Var (x_0 | \mathcal{Z}) &= \frac{1}{\frac{1}{qt} + \frac{1}{(1-q)(T-t)}+\frac{1}{q(T-t)}} + (1-q)t = \frac{1}{\frac{1}{t} + \frac{q}{1-q}\frac{1}{T}}
\end{align*}\n

\noindent as claimed.
\qed
\end{proof}

\subsection{Truth is Best Policy}

Does an expert gain or lose by deviating from the truth, reporting a prediction that is different from his actual belief? When allowed a single prediction, the fact that the logarithmic scoring rule is proper means that the expert's optimal policy is to predict truthfully. If allowed multiple predictions, this is not clear-cut:  A false prediction misdirects the market, so that a subsequent true prediction reaps the added benefit of correcting the misdirection. Plausibly, the gain of the latter outweighs the loss of the former.

We shall, however, show
\begin{proposition}
\label{truth}
If allowed periods when to predict are fixed, it is an equilibrium for the expert to make truthful predictions and for the market to take his predictions as truthful.
\end{proposition}

\begin{proof}

It is enough to show that if the market takes the expert's prediction as truthful, the expert cannot improve his expected reward by untruthfulness.

If the expert makes a single prediction, the result follows from the fact that the logarithmic scoring rule is proper. Therefore, assume more than one prediction, and let us focus on any consecutive pair of predictions.

Assume that an expert makes two consecutive predictions at times $T$ and $t < T$. The timing of the latter prediction at $t$ need not be known at $T$. It may depend on the expert's policy and signals revealed later than $T$. Given the expert's policy, at $T$, $t \in \Delta([T-1])$ is a random variable whose realization is conditioned on signals known after $T$.\footnote{It is here that the assumption that prediction periods are fixed is used. If not, $t$ may conditionally not exist.}

Let the market prediction at $T$ prior to an expert prediction be $X_{T-}\sim N(\mu_{T-},\sigma_{T-}^2)$. If the expert makes a truthful prediction, let the market's posterior prediction at $T$ be $X_{T+}\sim N(\mu_{T+},\sigma_{T+}^2)$, and the market's prior prediction at $t$ be $X_{t-}\sim N(\mu_{t-},\sigma_{t-}^2)$. If the expert predicts truthfully at $t$, let the market's posterior prediction at $t$ be $X_{t+}\sim N(\mu_{t+},\sigma_{t+}^2)$.

\begin{lemma}
\label{gain-loss}
Assume that the expert misrepresents his prediction mean at $T$ by an amount $c_T$ and that, as a result, the market's prior prediction mean at $t$, inferred from the (wrong) assumption that the expert is truthful, is distorted by $c_t$. Then the expert's net gain/loss expectation (at $T$) from the misrepresentation is
\s\begin{align*}
\E\limits_T[\Delta W] = \E\limits_T\Bigl[\frac{c_t^2}{2\sigma_{t-}^2}\Bigr] - \frac{c_T^2}{2\sigma_{T+}^2}
\end{align*}\n
\end{lemma}

\begin{proof}
By Propositions \ref{at-prediction} and \ref{after-prediction}, the expert's prediction variance is independent of any signal or predicted value. The expert's inferred variance is therefore unaffected by truthfulness.

We also observe in Proposition \ref{after-prediction} that $c_t$, the difference made to $\mu(x_0 |\mathcal{Z})$ by the expert declaring $y_T + c_T$ rather than $y_T$, is proportional to $c_T$ and is independent of any expert or market signal. For, by Propositions \ref{at-prediction} and \ref{after-prediction}
\s\begin{align*}
\sigma_{T+}^2 &= (1-q) T \\
\sigma_{t-}^2 &= \frac{1}{\frac{1}{t} + \frac{q}{1-q}\frac{1}{T}}
\end{align*}\n

\noindent and therefore
\s\begin{align}
\label{c_ratio}
c_t &= \frac{\frac{x_t}{t} + \frac{1}{1-q}\frac{y_T + c_T}{T} - \frac{x_T}{T}}{\frac{1}{t} + \frac{q}{1-q}\frac{1}{T}} - \frac{\frac{x_t}{t} + \frac{1}{1-q}\frac{y_T}{T} - \frac{x_T}{T}}{\frac{1}{t} + \frac{q}{1-q}\frac{1}{T}} = \frac{\sigma_{t-}^2}{\sigma_{T+}^2} c_T
\end{align}\n

By Propositions \ref{at-prediction} and \ref{after-prediction} again, the market's prediction is affected only by the expert prediction that immediately preceded it, and does not depend on any earlier predictions. Therefore a false prediction at $T$ has no effect on the expert's reward for any prediction made later than $t$ or earlier than $T$. Furthermore, for the two affected rewards at $T, t$, only some terms are affected: Referring to \eqref{W}, a false prediction at time $T$ affects only the term $- \frac{(x_0 - \mu_{+})^2}{2\sigma_{+}^2}$ at $T$, and only the term $+ \frac{(x_0 - \mu_{-})^2}{2\sigma_{-}^2}$ at $t$.

We calculate the expectation $\E\limits_T[\cdot]$ of the difference the false prediction makes to each of these two affected terms, i.e., the expectation at $T$ taken over the expert's true distribution.

\begin{itemize}
\item Difference to reward at $T$:

In the affected term $- \frac{(x_0 - \mu_{+})^2}{2\sigma_{+}^2}$, note that $\mu_{+} = \mu_{T+}$ if the expert predicted truthfully, while $\mu_{+} = \mu_{T+} + c_T$ if he lied about his mean.

The reward expectation difference, according to the expert's distribution at $T$, is therefore
\s\begin{align}
\label{diff-at-T}
\E\limits_T\Bigl[-\frac{(x_0 - \mu_{T+} - c_T)^2}{2 \sigma_{T+}^2}  + \frac{(x_0 - \mu_{T+})^2}{2 \sigma_{T+}^2}\Bigr] &= \frac{-c_T^2 + 2c_T\E\limits_T[x_0 - \mu_{T+}]}{2 \sigma_{T+}^2} \nonumber \\ 
&= - \frac{c_T^2}{2 \sigma_{T+}^2}
\end{align}\n

\item Difference to reward at $t$:

In the affected term $+ \frac{(x_0 - \mu_{-})^2}{2\sigma_{-}^2}$, note that $\mu_{-} = \mu_{t-}$ if the expert predicted truthfully, while $\mu_{-} = \mu_{t-} + c_t$ if he lied about his mean.

The reward expectation difference, according to the expert's distribution at $T$, is therefore
\s\begin{align}
\E\limits_T\Bigl[\frac{(x_0 - \mu_{t-} - c_t)^2}{2 \sigma_{t-}^2}  - \frac{(x_0 - \mu_{t-})^2}{2 \sigma_{t-}^2}\Bigr] = \nonumber \\
\E\limits_T\Bigl[\frac{c_t^2}{2 \sigma_{t-}^2}\Bigr] - \E\limits_T\Bigl[\frac{c_t (x_0 - \mu_{t-})}{\sigma_{t-}^2}\Bigr]
 \label{diff-at-t}
\end{align}\n

From \eqref{c_ratio}, $\frac{c_t}{\sigma_{t-}^2} = \frac{c_T}{\sigma_{T+}^2}$ is a non-random constant (given information known at $T$). Therefore \s$$\E\limits_T\Bigl[\frac{c_t (x_0 - \mu_{t-})}{\sigma_{t-}^2}\Bigr] = \frac{c_T}{\sigma_{T+}^2}\E\limits_T[x_0 - \mu_{t-}]$$\n

To evaluate $\E\limits_T[x_0 - \mu_{t-}]$, observe that $X_{t-}$ is a conditional distribution on $X_{T+}$. By the Law of Total Expectation
\s\begin{align*}
\E\limits_{T+}[\E\limits_{t-}[X_{t-}]] = \E\limits_{T+}[X_{T+}]
\end{align*}\n
I.e.,
\s\begin{align*}
\E\limits_{T+}[\mu_{t-}] = \E\limits_T[\mu_{t-}] = \mu_{T+}
\end{align*}\n

\noindent since $\E\limits_T[x_0] = \mu_{T+}$, we conclude $\E\limits_T[x_0 - \mu_{t-}] = 0$, and we get from \eqref{diff-at-t}
\s\begin{align}
\label{diff-at-tT}
\E\limits_T\Bigl[\frac{(x_0 - \mu_{t-} - c_t)^2}{2 \sigma_{t-}^2}  - \frac{(x_0 - \mu_{t-})^2}{2 \sigma_{t-}^2}\Bigr] =
\E\limits_T\Bigl[\frac{c_t^2}{2 \sigma_{t-}^2}\Bigr]
\end{align}\n
\end{itemize}

Adding \eqref{diff-at-T} and \eqref{diff-at-tT} the lemma follows.
\qed
\end{proof}

Fix any $0 < t < T$. Then by \eqref{c_ratio}
\s\begin{align*}
\frac{c_t^2}{2\sigma_{t-}^2} - \frac{c_T^2}{2\sigma_{T+}^2} = \frac{c_T^2}{2\sigma_{T+}^4}(\sigma_{t-}^2 - \sigma_{T+}^2) \leq 0
\end{align*}\n

\noindent because $\sigma_{T+}^2 / \sigma_{t-}^2 = (1-q)T/t + q \geq 1$.\footnote{Or, more generally, $\sigma_{t-}^2 \leq \sigma_{T+}^2$, because $X_{t-}$ is an inferred distribution from $X_{T+}$.}

By the lemma \s$$\E\limits_T[\Delta W] = \E\limits_T\Bigl[\frac{c_t^2}{2\sigma_{t-}^2} - \frac{c_T^2}{2\sigma_{T+}^2}\Bigr]$$\n
Since the expression under expectation is non-positive, so is the expectation, i.e., $\E\limits_T[\Delta W] \leq 0$. We conclude that for every $T$, for any expert policy, and independently of any other prediction the expert has made or will make, the reward expectation for distorting a prediction at $T$ is not positive. If the allowed periods for prediction are fixed and unaffected by the value of the expert's prediction, the expert maximizes his multi-prediction benefit expectation by making a truthful prediction at this particular prediction, and therefore at all predictions.

This completes the proof of Proposition \ref{truth}.
\qed
\end{proof}

Therefore, truthfulness is best policy for the expert. Note that if the expert's allowed prediction schedule is not fixed, the result may no longer be true. E.g., if the expert is allowed further predictions only if the discrepancy between his last prediction and the market's prediction exceeded some threshold, the expert may be motivated to distort his prediction so as to be given further prediction opportunities.

\subsection{Prediction Reward Expectation}

Having seen that there is no profit in lying, we shall from now on assume truthful predictions.

The following lemma will be useful in calculating the expected reward of a {\em future} prediction, before some of the signals it is based on are known. It shows that current and historic signals affect the reward expectation of the next prediction, but have no effect on the reward expectation of later predictions.

\begin{lemma}
\label{expect-before}
Assume that the expert is committed to making two consecutive predictions at $T$ and $t < T$. Let $X_{t-} \sim N(\mu_{-},\sigma_{-}^2)$ and $X_{t+} \sim N(\mu_{+},\sigma_{+}^2)$ be the market's distributions for $x_0$ before and after, respectively, a prediction $\mu_{+}$ is made at $t$.

Assume that $\sigma_{+}^2$ and $\sigma_{-}^2$ do not depend on any signals, but only on $T, t$ and $q$. Then, at any time $\tau \geq T$, the expected $t$-prediction reward is
\s\begin{align*}
\E_\tau[W_t] =  \log \frac{\sigma_{-}}{\sigma_{+}}
\end{align*}\n
\end{lemma}

\begin{proof}
In proving this lemma we shall note that, since the market is Bayesian, later distributions are conditional distributions on earlier ones. Specifically, they are conditional on signals and predictions which became known since. This enables us to use general laws such as the Law of Total Expectation and the Law of Total Variance which apply to conditional distributions.

We also use the result of Proposition \ref{at-prediction}, that, immediately after a prediction, market and expert distributions are identical.

\textbf{Throughout this lemma expectations are only over histories that conform to the lemma conditions, i.e., histories that include consecutive predictions at} $T, t$.

Since $X_{t+}$ is a conditional distribution on $X_{t-}$ ($X_{t+} = X_{t-} | \mu_{+}$), we may apply the Law of Total Variance
\s\begin{align}
Var_{t-}(X_{t-}) &= \E_{t-}[Var_{t+}(X_{t+})] + Var_{t-}(\E_{t+}[X_{t+}]) \nonumber \\
\label{totalvar}
&= \E_{t-}[Var_{t+}(X_{t+})] + \E_{t-}\Bigl[\bigl(\E_{t+}[X_{t+}] - \E_{t-}[\E_{t+}[X_{t+}]]\bigr)^2\Bigr]
\end{align}\n

We have $Var_{t-}(X_{t-}) = \sigma_{-}^2$, $\E\limits_{t+}[X_{t+}] = \mu_{+}$, and as $\sigma_{+}^2$ does not depend on $\mu_{+}$, $\E\limits_{t-}[Var(X_{t+})] = \sigma_{+}^2$. By the Law of Total Expectation $\E\limits_{t-}[\E\limits_{t+}[X_{t+}]] = \E\limits_{t-}[X_{t-}] = \mu_{-}$. Substituting in \eqref{totalvar}
\s\begin{align}
\label{conditional}
\sigma_{-}^2 = \sigma_{+}^2 + \E_{t-}[(\mu_{+} - \mu_{-})^2]
\end{align}\n

As $X_{t-}$ is conditional on $X_{T+}$, by the Law of Total Expectation $\E\limits_{T+}[\E\limits_{t-}[(\mu_{+} - \mu_{-})^2]] = \E\limits_{T+}[(\mu_{+} - \mu_{-})^2]$. So, taking expectations at $T+$ of \eqref{conditional}, and noting that variances are independent of signals
\s\begin{align}
\label{conditional2}
\sigma_{-}^2 = \sigma_{+}^2 + \E_{T+}[(\mu_{+} - \mu_{-})^2]
\end{align}\n

Since immediately after a truthful prediction, market and expert's distributions are identical, \eqref{conditional} is true not only from the market's point of view but also from the expert's. We may therefore substitute it in \eqref{E_W} (where expectations are taken over expert's distribution) and get
\s\begin{align*}
\E_T[W_t] =  -\frac{1}{2} \log \frac{\sigma_{+}^2}{\sigma_{-}^2}
\end{align*}\n

We mark the expectation $\E\limits_T[\cdot]$, as the distinction before and after prediction is not relevant to the expert.

It follows, using the lemma's assumption that $t, T, q$ are held constant, that $\E\limits_T[W_t]$ is constant. Therefore, for any $\tau \geq T$, by the Law of Total Expectation
\s\begin{align*}
\E_\tau[W_t]  = \E_\tau[\E_T[W_t]] = -\frac{1}{2} \log \frac{\sigma_{+}^2}{\sigma_{-}^2} = \log \frac{\sigma_{-}}{\sigma_{+}}
\end{align*}\n

\noindent as claimed.
\qed
\end{proof}

We remark here that it is easy to verify in the proof above, that the lemma also holds {\em a priori}, before the expert and market have received their first signal. Consequently by Proposition \ref{first} the {\em a priori} expected benefit of a first prediction is $\frac{1}{2} \log \frac{1}{1-q}$.

Another consequence is the following proposition
\begin{proposition}
\label{consecutive}
Assume that the expert is committed to making two consecutive predictions at $T$ and $t < T$. Then, at or before $T$, the reward expectation of the latter prediction is
\s\begin{align*}
\E[W_t] = -\frac{1}{2} \log \Bigl(1 - q\frac{T-t}{T}\Bigr)
\end{align*}\n
\end{proposition}

\begin{proof}
By Proposition \ref{after-prediction} $\sigma_{-}^2 = \frac{1}{\frac{1}{t} + \frac{q}{1-q}\frac{1}{T}}$, while by Proposition \ref{at-prediction} $\sigma_{+}^2 = (1 - q) t$.

These depend on $t, T$ and $q$ only, so by Lemma \ref{expect-before}, the reward expectation at every $\tau \geq T$ is
\s\begin{align*}
\E_\tau[W_t] &=  \log \frac{\sigma_{-}}{\sigma_{+}} = -\frac{1}{2} \log \Bigl( (1 - q) t\Bigl[\frac{1}{t} + \frac{q}{1-q}\frac{1}{T}\Bigr]\Bigr) \\ 
&= -\frac{1}{2} \log \Bigl(1 - q\frac{T-t}{T}\Bigr)
\end{align*}\n
\qed
\end{proof}

We can now prove the main result: The best strategy is to make predictions whenever allowed, and to make these predictions truthfully, as already shown (Proposition \ref{truth}).

\begin{theorem}
\label{predict-always}
If allowed periods for prediction are fixed, an expert maximizes his reward by making predictions at every allowed period, speaking the truth at all predictions.
\end{theorem}

Here is an overview of the proof, which follows below.

The proof is by induction on the number of allowed remaining periods for prediction, so that the induction assumption is that the expert will make truthful predictions in all $k-1$ last allowed periods, and the proof is completed by showing that the expert will also speak in the $k$'th last period. We therefore need to show that, at the $k$'th last period, and whatever signals he is observing, the expert gives higher reward expectation for making a prediction at this period (and all $k-1$ subsequent ones), than to staying silent, and making the next prediction at the $k-1$'th last period (and all subsequent ones).

The reward expectation for all but the next prediction is calculated using Proposition \ref{consecutive}. For the next prediction, if it takes place at the $k$'th last period, the expectation is computed from the expert's observed signals, while if it takes place at the $k-1$'th last period (i.e., in the future, relative to when expectation is computed), the expectation is computed based on the fact that a certain function of observed signals is a martingale. The total expectation is shown to be always greater when making a prediction at the $k$'th last period.

\begin{proof}
Let $K \subseteq [T_{max}]$ be the allowed periods for prediction.

We shall prove the theorem by induction. To this end we revise the theorem claim as follows:

\textbf{Claim:} For every $T$, an expert maximizes his reward by making truthful predictions at every period with $t < T, t \in K$ (regardless of which signals he receives before or after $T$, and regardless of his prediction history prior to $T$).

Clearly when $T = T_{max}$, this claim is equivalent to our theorem. We prove it by induction on the number of remaining allowed periods for prediction $n = n(T,K) := |K \cap [T]|$.

For $n = 1$, the theorem is true, because, by Proposition \ref{expected}, the expected reward for making each truthful prediction is non-negative. When $n = 1$, this is also the total expected reward, and so is weakly preferred to not making a prediction, a choice whose total reward is $0$.

As the induction hypothesis, we assume the claim true for $n-1$. We proceed to prove it for $n$:

Let $A \cap [T] = \{t_1, t_2, \ldots, t_n\}$ with $t_1 > t_2 > \ldots > t_n$. By the induction hypothesis, the expert will make predictions at all periods $K \setminus t_1$.  We must show that, for all histories at $t_1$, an expert who is informed of $x_{t_1}, y_{t_1}$ prefers to make his next prediction immediately over making it at $t_2$.

We treat two cases

\begin{enumerate}[(i)]
\item The expert has not already made a prediction, i.e., $t_1, t_2$ are the earliest periods in $K$.

\begin{lemma}
\label{at-t}
At $T$, for every $t \leq T$
\s\begin{align*}
\E_T[(x_t - y_t)^2] = \frac{t^2}{T^2}(x_T - y_T)^2 + q\frac{t(T-t)}{T}
\end{align*}\n
\end{lemma}

\begin{proof}
According to the model
\s\begin{align*}
0 &= (x_t - y_t) - \sum\limits_{i=1}^t A_i \\
x_T - y_T &= (x_t - y_t) + \sum\limits_{i=t+1}^T A_i
\end{align*}\n

Since $- \sum\limits_{i=1}^t A_i$ and $\sum\limits_{i=t+1}^T A_i$ are mutually independent, we have two independent observations ($0$ and $x_T - y_T$) of the random normal variable $X_t - Y_t$, with variances $qt$ and $q(T-t)$, respectively, which therefore has posterior distribution $N(\mu,\sigma^2)$ with
\s\begin{align*}
\mu &= \frac{\frac{0}{qt} + \frac{x_T - y_T}{q(T-t)}}{\frac{1}{qt} + \frac{1}{q(T-t)}} = \frac{t}{T}(x_T - y_T) \\
\sigma^2 &= \frac{1}{\frac{1}{qt} + \frac{1}{q(T-t)}} = q\frac{t(T-t)}{T}
\end{align*}\n

As the second moment of $N(\mu,\sigma^2)$ is $\mu^2 + \sigma^2$, the lemma follows.
\footnote{An alternative line of proof would be to show that $\frac{x_t - y_t}{t}$ is a martingale. We use this alternative method in the other case.}
\qed
\end{proof}

To complete the proof, it is sufficient to show that speaking first at $t_2$ does not carry a higher expected reward than speaking first at $t_1$, with reward expectations taken at $t_1$.

By Propositions \ref{first} and \ref{consecutive} and the induction hypothesis, the optimal reward expectation of an expert whose first prediction is at $t$, which we denote by $W_{\leq t}$, is
\s\begin{align}
\E_t[W_{\leq t}] &= \frac{1}{2}\Bigl[\frac{(x_t - y_t)^2}t - q + \log(1-q) - \sum_{t_i < t} \log(1 - q\frac{t_{i-1} - t_i}{t_{i-1}})\Bigr]
\label{expect-after-t}
\end{align}\n

By \eqref{expect-after-t} and Lemma \ref{at-t}, the expert's expectation of $W_{\leq t_2}$ at $t_1$ is
\s\begin{align}
\label{expect-after-t-at-T}
\E_{t_1}[W_{\leq t_2} | x_{t_1}, y_{t_1}] &= \frac{1}{2}\Bigl[\frac{t_2}{t_1^2}(x_{t_1} - y_{t_1})^2 + q\frac{t_1-t_2}{t_1} \nonumber \\ &
- q + \log(1-q) - \sum_{t_i < t_2} \log(1 - q\frac{t_{i-1} - t_i}{t_{i-1}})\Bigr]
\end{align}\n

Now we show that $\E\limits_{t_1}[W_{\leq t_1}] \geq \E\limits_{t_1}[W_{\leq t_2} | x_{t_1}, y_{t_1}]$. By \eqref{expect-after-t-at-T} and substituting $t_1$ for $t$ in \eqref{expect-after-t}
\s\begin{align*}
\E_{t_1}[W_{\leq t_1}] &- \E_{t_1}[W_{\leq t_2} | x_{t_1}, y_{t_1}] \\ 
&= \frac{1}{2}\Bigl[\frac{t_1 - t_2}{t_1^2}(x_{t_1} - y_{t_1})^2 - q\frac{t_1 - t_2}{t_1} - \log(1 - q\frac{t_1 - t_2}{t_1})\Bigr]  \\
&>= \frac{1}{2}\Bigl[ - q\frac{t_1 - t_2}{t_1} - \log(1 - q\frac{t_1 - t_2}{t_1})\Bigr] \\
&\geq 0
\end{align*}\n

Because for every $x < 1$, $\log (1-x) \leq -x$.

\item The expert has already made a prediction.

Suppose the expert's last prediction was at $T$. If his next prediction is at $t$, his reward expectation (at $t$) for all subsequent predictions is, from \eqref{E_W}, Proposition \ref{consecutive} and the induction hypothesis
\s\begin{align}
\E_t[W_{\leq t}] &= \frac{1}{2}\Bigl[\frac{(x^*_t - y_t)^2}{\sigma_{t-}^2} + \Bigl(\frac{\sigma_{t+}^2}{\sigma_{t-}^2} - 1 - \log \frac{\sigma_{t+}^2}{\sigma_{t-}^2}\Bigr)\nonumber\\ 
& - \sum_{t_i < t} \log(1 - q\frac{t_{i-1} - t_i}{t_{i-1}})\Bigr]
\label{W_t}
\end{align}\n

\noindent where, by Propositions \ref{at-prediction} and \ref{after-prediction},
\s\begin{align}
\label{x*}
x^*_t &:= \frac{\frac{x_t}{t} + \frac{1}{1-q}\frac{y_T}{T} - \frac{x_T}{T}}{\frac{1}{t} + \frac{q}{1-q}\frac{1}{T}}\\
\label{sigma}
\sigma_t^2 &:= \sigma_{t-}^2 :=\frac{1}{\frac{1}{t} + \frac{q}{1-q}\frac{1}{T}} \\
\sigma_{t+}^2 &:= (1 - q) t
\end{align}\n

So that

\s\begin{align}
\sigma_{t+}^2 /  \sigma_{t-}^2 &= 1 - q\frac{T-t}{T}
\end{align}\n

We wish to show that, for every $t_1 > t_2$, the expert, given his signals at $t_1$, prefers making a prediction immediately (at $t_1$), rather than later (at $t_2$), i.e., that \s$$\E\limits_{t_1}[W_{\leq t_1}] \geq \E\limits_{t_1}[W_{\leq t_2} | x_{t_1}, y_{t_1}]$$\n

To calculate $\E\limits_{t_1}[W_{\leq t_2} | x_{t_1}, y_{t_1}]$, the later expectation given present signals, we must derive a distribution for $(x^*_{t_2} - y_{t_2}) | x_{t_1}, y_{t_1}$. The following lemma answers this.

\begin{lemma}
\label{Z_t_2}
For every $\tau < t$ the distribution of $Z_\tau := (x^*_\tau - y_\tau) | x_t, y_t, x_T, y_T$ is normal with
\s\begin{align}
\label{mu_Z}
\mu(Z_\tau) &= \frac{\sigma_\tau^2}{\sigma_t^2}(x^*_t - y_t) \\
\label{var_Z}
Var(Z_\tau) &= q(t - \tau)\Bigl[\frac{1}{t \tau} + \frac{q}{(1-q)T^2}\Bigr] \sigma_\tau^4
\end{align}\n
\end{lemma}

\begin{proof}
We show that $\frac{x^*_t - y_t}{\sigma_t^2}$ is a martingale. From \eqref{x*} and \eqref{sigma}
\s\begin{align}
\nonumber
\frac{x^*_t - y_t}{\sigma_t^2} &= \frac{x_t - y_t}{t} + \frac{1}{1-q}\frac{y_T - y_t}{T} - \frac{x_T - y_t}{T} \\
\nonumber
&=  \frac{x_t - y_t}{t} + \frac{q}{1-q}\frac{y_T - y_t}{T} - \frac{x_T - y_T}{T} \\
\label{diffxy}
&=  \frac{1}{t} \sum\limits_{i=1}^t A_i + \frac{q}{(1-q)T}\sum\limits_{i= t + 1}^T B_i - \frac{x_T - y_T}{T}
\end{align}\n

Since the $A_i$'s are i.i.d., by symmetry for every $1 \leq i \leq t$, $\E[A_i | x_t, y_t, x_T, y_T] = \frac{x_t - y_t}{t}$. Clearly $\E[B_i | x_t, y_t, x_T, y_T] = 0$, since knowing $x_t, y_t, x_T, y_T$ provides no information on the distribution of the $B_i$'s for $i \leq t$. Conditioning  \eqref{diffxy} on $x_t, y_t, x_T, y_T$, we see that for every $1 \leq \tau \leq t$
\s\begin{align}
\label{nice}
\E_t\Bigl[\frac{Z_\tau}{\sigma_\tau^2}\Bigr] = \frac{Z_t}{\sigma_t^2}
\end{align}\n

\noindent from which \eqref{mu_Z} follows.

From \eqref{diffxy} we have
\s\begin{align}
\label{diff}
\frac{x^*_\tau - y_\tau}{\sigma_\tau^2} - \frac{x^*_t - y_t}{\sigma_t^2} &=  \Bigl(\frac{1}{\tau} - \frac{1}{t}\Bigr) \sum\limits_{i=1}^\tau A_i - \frac{1}{t} \sum\limits_{i=\tau + 1}^t A_i + \frac{q}{(1-q)T}\sum\limits_{i= \tau + 1}^t B_i
\end{align}\n

By the Law of Total Variance
\s\begin{align*}
Var\Bigl(\frac{Z_\tau}{\sigma_\tau^2}\Bigr) & = Var\Bigl(\frac{x^*_\tau - y_\tau}{\sigma_\tau^2}\Bigr) - Var\Bigl( \frac{x^*_t - y_t}{\sigma_t^2}\Bigr) \\
&= Var\Bigl(\frac{x^*_\tau - y_\tau}{\sigma_\tau^2} - \frac{x^*_t - y_t}{\sigma_t^2}\Bigr) 
\end{align*}\n

And so, since the terms of the right-hand side of \eqref{diff} are mutually independent, we have
\s\begin{align*}
Var\Bigl(\frac{Z_\tau}{\sigma_\tau^2}\Bigr) & = Var\Bigl(\frac{x^*_\tau - y_\tau}{\sigma_\tau^2} - \frac{x^*_t - y_t}{\sigma_t^2}\Bigr) \\ 
&= \Bigl(\frac{1}{\tau} - \frac{1}{t}\Bigr)^2 q \tau + \frac{1}{t^2} q (t - \tau) + \frac{q^2}{(1-q)^2 T^2} (1-q) (t - \tau) \\
&= q(t - \tau)\Bigl[\frac{1}{t \tau} + \frac{q}{(1-q)T^2}\Bigr]
\end{align*}\n

\noindent from which \eqref{var_Z} follows.
\qed
\end{proof}

Using Lemma \ref{Z_t_2} and \eqref{W_t}, let us evaluate $\E\limits_{t_1}[W_{\leq t_1}] - \E\limits_{t_1}[W_{\leq t_2} | x_{t_1}, y_{t_1}]$, the difference between speaking at $t_1$ and speaking at $t_2$.
\s\begin{align*}
2 \bigl(\E_{t_1}&[W_{\leq t_1}] - \E_{t_1}[W_{\leq t_2} | x_{t_1}, y_{t_1}]\Bigr) \\
&= \frac{1}{\sigma_{t_1}^2} \Bigl(1 - \frac{\sigma_{t_2}^2}{\sigma_{t_1}^2} \Bigr) (x^*_{t_1} - y_{t_1})^2  - q(t_1 - t_2)\Bigl[\frac{1}{t_1 t_2} + \frac{q}{(1-q)T^2}\bigr] \sigma_{t_2}^2 \\
& + \bigl(1- q\frac{T-t_1}{T}\bigr) - 1 - \log\bigl(1 - q\frac{T-t_1}{T}\bigr) \\
& - \bigl(1 - q\frac{T-t_2}{T}\bigr) + 1 + \log\bigl(1 - q\frac{T-t_2}{T}\bigr) \\
& - \log\bigl(1 - q\frac{t_1 - t_2}{t_1}\bigr) \\
&= \frac{1}{\sigma_{t_1}^2} \Bigl(1 - \frac{\sigma_{t_2}^2}{\sigma_{t_1}^2} \Bigr) (x^*_{t_1} - y_{t_1})^2  \\
&+ q\frac{t_1 - t_2}{T} - q(t_1 - t_2)\Bigl[\frac{1}{t_1 t_2} + \frac{q}{(1-q)T^2}\Bigr] \sigma_{t_2}^2 \\
& - \log\frac{\bigl(1 - q\frac{T-t_1}{T}\bigr) \bigl(1 - q\frac{t_1 - t_2}{t_1}\bigr)}{1 - q\frac{T-t_2}{T}}
\end{align*}\n

Since $\sigma_t^2$ is increasing in $t$, we have $1 - \frac{\sigma_{t_2}^2}{\sigma_{t_1}^2} > 0$, and we get from the above the following inequality
\s\begin{align}
\label{ineq}
2 \bigl(\E_{t_1}&[W_{\leq t_1}] - \E_{t_1}[W_{\leq t_2} | x_{t_1}, y_{t_1}]\bigr) \nonumber \\
&> q\frac{t_1 - t_2}{T} - q(t_1 - t_2)\Bigl[\frac{1}{t_1 t_2} + \frac{q}{(1-q)T^2}\Bigr] \sigma_{t_2}^2 \nonumber \\
& - \log\frac{\bigl(1 - q\frac{T-t_1}{T}\bigr) \bigl(1 - q\frac{t_1 - t_2}{t_1}\bigr)}{1 - q\frac{T-t_2}{T}}
\end{align}\n

Define
\s\begin{align}
\label{R}
R &:= -q\frac{t_1 - t_2}{T} + q(t_1 - t_2)\Bigl[\frac{1}{t_1 t_2} + \frac{q}{(1-q)T^2}\Bigr] \sigma_{t_2}^2 \\
\label{S}
S &:= 1 - \frac{\bigl(1 - q\frac{T-t_1}{T}\bigr) \bigl(1 - q\frac{t_1 - t_2}{t_1}\bigr)}{1 - q\frac{T-t_2}{T}}
\end{align}\n

\noindent so that \eqref{ineq} can be rewritten
\s\begin{align*}
2 \bigl(\E_{t_1}[W_{\leq t_1}] - \E_{t_1}[W_{\leq t_2} | x_{t_1}, y_{t_1}]\bigr) &> -R - \log(1 - S)
\end{align*}\n

Simplifying and rearranging \eqref{R} and \eqref{S}, one discovers that $R = S$. Therefore, as for every $x < 1$, $\log (1-x) \leq -x$, we conclude
\s\begin{align*}
\E_{t_1}[W_{\leq t_1}] - \E_{t_1}[W_{\leq t_2} | x_{t_1}, y_{t_1}] > 0
\end{align*}\n

\noindent which settles this case.
\end{enumerate}
\qed
\end{proof}

Consequently if the expert is allowed to speak every period, he will. The following proposition gives his reward expectation and its asymptotic behavior for large $T$.

\begin{theorem}
\label{average}
Assume the expert is allowed to make a prediction every period. Mark by $\Xi(T)$ the average reward expectation at period $T$ for an expert using optimal prediction strategy.
\s\begin{align}
\label{sum_logs}
\Xi(T) &= \frac{1}{2}\sum\limits_{t=1}^T \log \frac{t}{t-q} \\
\label{gamma_formula}
&= \frac{1}{2}\log \frac{\Gamma(1 - q)\Gamma(T + 1)}{\Gamma(T + 1 - q)}
\end{align}\n

For large $T$, $\Xi(T) = O(\log T)$. More specifically
\s\begin{align}
\label{limit}
\lim\limits_{T \to \infty} \frac{\Xi(T)}{q \log T + \log\Gamma(1-q)} = \frac{1}{2}
\end{align}\n
\end{theorem}

\begin{proof}
The optimal policy is to predict at every period starting at $T$. The expected reward (averaged over all random walks) for a first prediction at $T$ is, by Lemma \ref{expect-before}, $\frac{1}{2}\log \frac{1}{1-q}$, while for every $t < T$, it is, by Proposition \ref{consecutive}, $\frac{1}{2}\log \frac{t+1}{t+1-q}$. \eqref{sum_logs} follows, and from it \eqref{gamma_formula}.

We use the following limit of the Gamma function: For $\alpha \in \mathbb{R}$
\s\begin{align*}
\lim_{n \to \infty} \frac{\Gamma(n+\alpha)}{\Gamma(n) n^\alpha} = 1
\end{align*}\n

\eqref{limit} follows from this and \eqref{gamma_formula} by substituting $\alpha = -q, n = T + 1$.
\qed
\end{proof}

\section{Discussion}
\label{discussion}

\subsection{Conclusions}

We analyzed the expert's policy in the prediction scenario described, and found that an expert should make a new, truthful prediction whenever he is allowed to and has an updated signal. For large $t$, his total reward is on average roughly $\frac{1}{2} q \log t$. This compares to the asymptotic average reward of $q \log \log t$ that \cite{azar2016should} found is achievable by best policy of the expert when restricted to a single prediction.

The ability to revise predictions therefore significantly increases the expert's reward, by a factor $O(\frac{\log t}{\log \log t})$.

The results also show the answer to a related question: Suppose the market wishes to hire an expert to publicly weigh in on the market price for the duration of the prediction period. The results show that compensating the expert with the logarithmic scoring rule is incentive enough to achieve this goal, and the expected expense is proportional to the expert's expected reward.

\subsection{Other Random Walks}

In Section \ref{results} we noted that our main results, and particularly the truthfulness property, does not necessarily apply to any other model. It is an interesting open problem to characterize which models do, in fact, lead to similar results.
It needs reminding that our derivations depend on two critical elements of our model:
\begin{enumerate}
\item Gaussian random walk
\item Variances of all signals and, as a result, of inferred distributions are common knowledge, and consequently independent of signal values.
\end{enumerate}

No similar results may apply, for example, in binary prediction markets (where there is a 0/1 outcome), because in the underlying Bernoulli distribution, a prediction $p$, representing the distribution's mean, also affects the distribution variance $p(1-p)$.

\subsection{Informational Substitutes}

\cite{chen2016informational} formulate a criterion of ``informational substitutes'' as leading to being truthful and revealing information at first opportunity, in a world where agents' private information is static. The definition used for ``informational substitutes'' is that information is more valuable (per the scoring rule in force) earlier than later. While simply stated, working this out for any given case may be involved.

We find the same, in our world where private information is dynamic. In that context, the major part of the proof Theorem \ref{predict-always} was to show that expert and market's signals are informational substitutes.

Our result is therefore consistent with a generalization of \cite{chen2016informational} (and \cite{Chen2007}) to dynamic-information contexts. We venture a guess that such a generalization will prove to be valid. As our analysis shows, such a generalization is not straight-forward, but depends, {\em inter alia}, on the behavior of interdependent martingales.

\subsection{Several Experts}

Our results hold for a single expert. With more than one expert, neither truthfulness nor promptness in revealing information is guaranteed, as we already know from static-information studies. It would be useful to have a criterion to separate the truthful scenarios from the untruthful ones. As said earlier, we believe that a generalization of \cite{chen2016informational}'s ``informational substitutes'' to dynamic-information structures will prove to be such a criterion, even though working out that criterion for any given setting might remain a non-trivial undertaking.

Where truthfulness and promptness are found to be optimal, the remaining question would be the expected reward, the equivalent of our Theorem \ref{average}. Tie-breaking rules would be needed for experts predicting simultaneously. The logical generalization of Theorem \ref{average} is that each expert's expected reward will be asymptotically $Q \log t$, where $Q$ is a measure of the expert's marginal information per unit time ($Q = \frac{1}{2}q$ in the single-expert scenario).

\bibliographystyle{ACM-Reference-Format-Journals}
\bibliography{infomarkets}

\end{document}